\documentclass[submission]{eptcs}

\providecommand{\event}{ITRS 2014}

\usepackage[T1]{fontenc}
\usepackage[utf8]{inputenc}
\usepackage[english]{babel} 
\usepackage{mathpartir}
\usepackage{amsmath}

\usepackage[amsmath]{ntheorem}
\usepackage{amssymb}
\usepackage{stmaryrd}
\usepackage{listings}

\newcommand{\judgment}[2]{\ensuremath{\Gamma\vdash^\cap_{\mathbb{Q}}
    #1 : #2}}

\newenvironment{types}[2][c]
{\begin{minipage}[#1]{#2}\centering
\hrule height2pt
\centerline{\vrule width2pt height 5pt\hfill \vrule width2pt height 5pt}
\begin{minipage}{\dimexpr\textwidth-4pt-1em}}
{\end{minipage}
\centerline{\vrule width2pt height 5pt\hfill \vrule width2pt height 5pt}
\hrule height2pt
\end{minipage}}

\newenvironment{proof}[1][Proof]{\begin{trivlist}
\item[\hskip \labelsep {\itshape #1.}]}{\end{trivlist}}

\newcommand{\qed}{\hfill$\square$}

\newtheorem{lemma}{Lemma}
\newtheorem{corollary}{Corollary}

\newtheorem{theorem}{Theorem}

\theoremstyle{definition}
\newtheorem{mydef}{Definition}

\title{Liquid Intersection Types}
\author{Mário Pereira
\quad
Sandra Alves
\quad
Mário Florido
\institute{University of Porto, Department of Computer Science \&
 LIACC}
\email{\{mariopereira,sandra,amf\}@dcc.fc.up.pt}
}

\def\titlerunning{Liquid Intersection Types}
\def\authorrunning{M. Pereira, S. Alves \& M. Florido}

\begin{document} 
\maketitle

\begin{abstract}
  We present a new type system combining refinement types and
  the expressiveness of intersection type discipline. The use of such
  features makes it possible to derive more precise types than in the
  original refinement system. We have been able to prove several
  interesting 
  properties for our system (including subject reduction) and developed
  an inference algorithm, which we proved to be
  sound. 
\end{abstract}

\section{Introduction}
\label{sec:intro}

Refinement types
\cite{Knowles:2010:HTC:1667048.1667051} state
complex program 
invariants, by augmenting type systems with logical
predicates. A refinement type of the form $\left\{\nu : B \: | \:
  \phi\right\}$ stands for the set of values from basic type $B$
restricted to the filtering predicate (refinement) $\phi$.
A subtyping relation exists for refinement types, which will generate
implication conditions: 
\[
\inferrule
{
  \Gamma;\nu : B \vdash \phi\Rightarrow\psi
}
{
  \Gamma\vdash \{\nu: B \: | \: \phi\} <: \{\nu : B \: | \: \psi\}
}
\]

One idea behind the use of such type systems is to
perform 
type-checking using SMTs (Satisfability Modulo
Theories) solvers~\cite{Shostak:1984:DCT:2422.322411}, discharging 
conditions as the above $\phi \Rightarrow \psi$. 
However,
the use of arbitrary boolean terms as refinement expressions leads to
undecidable type systems, both for type checking and inference. 

Liquid Types \cite{Rondon:2008:LT:1375581.1375602,
  Vazou:2013:ART:2450268.2450286} present a system capable of
automatically 
inferring refinement types, by means of two main restrictions to a general 
refinement type system: refinement predicates of some terms are
conjunctions of  
expressions exclusively taken from a global, user-supplied set
(denoted $\mathbb{Q}$) of logical qualifiers (simple predicates over
program variables, the value variable $\nu$ and the variable
placeholder $\star$); and a conservative (hence decidable) notion of
subtyping. 


Despite the interest of Liquid Types, some situations arise where the
inference procedure infers poorly accurate types. For example,
considering 
$\mathbb{Q} = \{\nu \geq 0, \nu \leq 0 \}$ and the term 
$\mathit{neg}\equiv \lambda x. - x$, Liquid Types infer for 
$\mathit{neg}$ the type $x:\{0 \leq \nu \wedge 0 \geq
\nu\}\rightarrow \{0 \leq \nu \wedge 0 \geq \nu\}$
(throughout this paper we write $\left\{\phi\right\}$ instead of
$\left\{\nu:B \: | \: \phi\right\}$ whenever $B$ is clear from the
context). This type cannot 
 be taken as a precise description of the \textit{neg}
function's behavior, since it is not expressed that for a positive
(resp. negative)
argument the function returns a negative (resp. positive) value. With
our system we will have for \textit{neg} 
the type $(x:\left\{\nu\geq 0\right\}\rightarrow\left\{\nu\leq
  0\right\})\cap(x:\left\{\nu\leq 0\right\}\rightarrow\left\{\nu\geq
  0\right\})$.

We introduce \emph{Liquid Intersection Types}, 
a refinement type system 
with the addition of intersection types
\cite{barendregt1983filter,coppo1980}. 
Our use of intersections in conjunction with refinement types is
motivated by a problem clearly identified for Liquid Types: the
absence of most-general types, as in the ML tradition. Our use of
intersections for refinement types draws inspiration from
\cite{Freeman:1991:RTM:113445.113468}, since this offers a way to use
jointly detailed types and intersections. Though, integrating this
expressiveness with refinement types and keeping the qualifiers from
$\mathbb{Q}$ simple (which must be provided by the
programmer) implies the design of a new type system. 

Besides the new type system, another contribution of this work is a new
inference 
algorithm for Liquid Intersection Types. 

This paper is organized as follows. Section \ref{sec:type-system}
presents the designed type system, with a focus on the language
syntax, semantics and typing rules, as well as a soundness result. The
type inference algorithm is introduced in section
\ref{sec:inference}. Finally, in section \ref{sec:conclusion} we
conclude with final remarks and explain some possible future work.



\section{Type system}
\label{sec:type-system}


\subsection{Syntax and semantics}
\label{sec:syntax-semantics}
\vspace{-10pt}
\begin{figure}
\small
\[
  \begin{array}{lrlr}
    M, N \qquad & ::= & & \quad \textit{Terms:} \\
    & | & x & \textit{variable} \\
    & | & \mathrm{c} & \textit{constant} \\
    & | & \lambda x.M & \textit{abstraction} \\
    & | & MN & \textit{application} \\
    & | & \mathrm{let} \; x = M \; \mathrm{in} \; N \quad &
    \textit{let-binding} \\
    & | & [\Lambda\alpha]M & \textit{type abstraction} \\
    & | & [\tau]M & \textit{type instantiation} \\
    \phi \qquad & ::= & & \textit{Liquid refinements:} \\
    & | & q & \textit{qualifier from $\mathbb{Q}$} \\
    & | & \top & \textit{true  (empty  refinement)} \\
    B \qquad & ::= & & \textit{Base types}: \\
    & | & \mathtt{int} &  \textit{integers}\\
    & | & \mathtt{bool} & \textit{booleans} \\
    \overset{\sim}{\tau}(\mathcal{R}) \qquad & ::= & & \textit{Pretype
      skeleton} \\ 
    & | & \left\{\nu : B \: | \: \mathcal{R} \right\} & \textit{base
      refined  
      type} \\ 
    & | &
    x:\overset{\sim}{\tau}(\mathcal{R})\rightarrow
    \overset{\sim}{\tau}(\mathcal{R})  
    & 
    \textit{function} \\ 
    & | &
    \overset{\sim}{\tau}(\mathcal{R})\cap\overset{\sim}{\tau}(\mathcal{R})
    & 
    \textit{intersection} \\ 
    & | & \alpha & \textit{type variable} \\
    \overset{\sim}{\sigma}(\mathcal{R}) \qquad & ::= & & \textit{Pretype
      scheme skeleton}:\\ 
    & | & \overset{\sim}{\tau}(\mathcal{R}) & \textit{mono pretype} \\
    & | & \forall\alpha.\overset{\sim}{\sigma}(\mathcal{R}) &
    \textit{pretype scheme}   
    \\
    T \qquad & ::= & & \textit{Simple types:} \\
    & | & B & \textit{basic type} \\
    & | & \alpha & \textit{type variable} \\
    & | & T_1\rightarrow T_2 & \textit{functional type} \\
    \overset{\mathbf{.}}{\tau}(\mathcal{R}),
    \overset{\mathbf{.}}{\sigma}(\mathcal{R}) 
    \qquad & ::= & 
    \overset{\sim}{\tau}(\mathcal{R}) :: T,  
    \overset{\sim}{\sigma}(\mathcal{R}) :: T & \textit{Well-founded
      pretype, scheme} 
    \\
    \tau, \sigma
    \qquad & ::= & 
    \overset{\mathbf{.}}{\tau}(E),  
    \overset{\mathbf{.}}{\sigma}(E) & \textit{Refinement
      Intersection Type, Scheme} 
    \\
    \hat{\tau}, \hat{\sigma} \qquad & ::= &
    \overset{\mathbf{.}}{\tau}(\phi),  
    \overset{\mathbf{.}}{\sigma}(\phi) & \textit{Liquid
      Intersection Type, Scheme} 
    \\
    \Gamma \qquad & ::= & & \textit{Environment:} \\
    & | & \emptyset & \textit{empty}\\
    & | & \Gamma;x:\sigma & \textit{new  binding}
  \end{array}
  \]
  \caption{Syntax}
  \label{fig:syn}
\end{figure}

Our target language is the $\lambda$-calculus extended with constants
and, as in the Damas-Milner type system, local bindings via the
\texttt{let} constructor. We assume the Barendregt convention
regarding names of free and bound variables \cite{bar84}, and identify
terms modulo $\alpha$-equivalence. The syntax of expressions and types
is presented in Figure \ref{fig:syn}. We will use $FV(M)$ and $BV(M)$
to denote the set of free and bound variables of term $M$,
respectively. These notions can be lifted to type environments, as
$FV(\Gamma)$, resp. $BV(\Gamma)$, denoting the free variables,
resp. the bound variables, of refinement expressions for every typed
bound within $\Gamma$.

The set of constants of our language is a countable alphabet of
constants $c$, including literals and primitive functions. We assume
for primitive functions the existence of at least arithmetic
operators, a fixpoint combinator \texttt{fix} and an identifier
representing \texttt{if-then-else} expressions. The type of constants
is established using a mapping $ty(c)$, assigning a refined type that
captures the semantic of each constant. For instance, to an integer
literal $n$ it would be assigned the type $\left\{\nu : \mathtt{int}
  \: | \: \nu = n\right\}$. Note that refinements may come from the
user defined set $\mathbb{Q}$ or from the constants and
sub-derivations. In the latter case the refinement expressions are
arbitrary expressions from $E$.


We use $\overset{\sim}{\tau}(\mathcal{R})$ and
$\overset{\sim}{\sigma}(\mathcal{R})$ to denote pretypes and pretype
schemes, respectively (this notion 
of pretypes goes back to \cite{Ong:2012:TGS:2359888.2359957}), which 
stand for type variables, basic and functional refined types,
intersection of pretypes and polymorphic pretypes. The notation
$x:\tau_1\rightarrow\tau_2$ will be
preferred over the usual $\Pi(x:\tau_1).\tau_2$ for functional
dependent types, meaning that variable $x$ may occur in the
refinement expressions present in $\tau_2$. 
An intersection in
pretypes (denoted by $\cap$) indicates that a term with type
$\overset{\sim}{\tau}_1(\mathcal{R})\cap
\overset{\sim}{\tau}_2(\mathcal{R})$  
has both type $\overset{\sim}{\tau}_1(\mathcal{R})$ and
$\overset{\sim}{\tau}_2(\mathcal{R})$, respecting the possible refinement
predicates figuring in these types. We assume the $'\cap'$ operator to be
commutative, associative and idempotent. 

A \emph{well-founded pretype} (resp. \emph{well-founded type scheme})
is a pretype 
$\overset{\sim}{\tau}(\mathcal{R})$
(resp. $\overset{\sim}{\sigma}(\mathcal{R})$) for such that
$\overset{\sim}{\tau}(\mathcal{R}) ::  
T$ (resp. $\overset{\sim}{\sigma}(\mathcal{R}) ::  
T$), for some $T$ ($T$ stands for simple types for the rest of
this document). The
\emph{well-founded} relation $\overset{\sim}{\sigma}(\mathcal{R}) :: T$ is
inductively defined by: 
\[
\small
\begin{array}{ccccc}
  \inferrule*[Lab=\footnotesize::-Var]
  {
  }
  {
    \alpha :: \alpha
  }  & 
  \inferrule*[Lab=\footnotesize::-Fun]
  {
    \overset{\sim}{\tau}_x(\mathcal{R}) :: T_x \\
    \overset{\sim}{\tau}(\mathcal{R}) :: T
  }
  {
    (x:\overset{\sim}{\tau}_x(\mathcal{R})\rightarrow
    \overset{\sim}{\tau}(\mathcal{R})) :: 
    T_x\rightarrow T
  }  & 
  \inferrule*[Lab=\footnotesize::-Ref]
  {
  }
  {
    \left\{\nu : B \: | \: \mathcal{R}\right\} :: B 
  }  & 
  \inferrule*[Lab=\footnotesize::-$\forall$]
  {
    \overset{\sim}{\sigma}(\mathcal{R}) :: T
  }
  {
    \forall\alpha.\overset{\sim}{\sigma}(\mathcal{R}) :: T
  }  & 
  \inferrule*[Lab=\footnotesize::-$\cap$]
  {
    \overset{\sim}{\tau}_1(\mathcal{R}) :: T \\
    \overset{\sim}{\tau}_2(\mathcal{R}) :: T
  }
  {
    \overset{\sim}{\tau}_1(\mathcal{R})\cap
    \overset{\sim}{\tau}_2(\mathcal{R}) :: T
  }
\end{array}
\]
Using this relation guarantees that
intersection of types are at the refinement expressions only,
i.e. for $\overset{\sim}{\sigma}_1(\mathcal{R})
\cap\overset{\sim}{\sigma}_2(\mathcal{R})$ both 
$\overset{\sim}{\sigma}_1(\mathcal{R})$ and
$\overset{\sim}{\sigma}_2(\mathcal{R})$ are of  
the same form, solely differing in the refinement predicates.


To describe the execution behavior of our language we use a small-step
contextual operational semantics, whose rules are shown in Figure
\ref{fig:op}.  The relation $M \leadsto N$ describes a single
evaluation step from term $M$ to $N$. The rules $[\mathcal{E}-\beta]$,
$[\mathcal{E}-Let]$ and $[\mathcal{E}-Compat]$ are standard for a
call-by-value ML-like language. The rule
$[\textit{$\mathcal{E}$-Constant}]$ evaluates an 
application with a constant in the function position. This rule relies
on the embedding $\left\llbracket\cdot\right\rrbracket$ of terms into
a decidable logic \cite{Nelson:1980:TPV:909447} (the definition of
this embedding, as well as the details of the used logic, will be made
clear in next section).


\begin{figure}
\small
  \[
  \begin{array}{rrlr}
    V & ::= & & \textit{Values:} \\
    & | & \mathrm{c} & \textit{constant} \\
    & | & \lambda x.M & \textit{abstraction} \\
    \mathrm{\mathbf{Contexts}} & & & \boxed{C} \\
    C & ::= & & Contexts: \\
    & | & [\;] & \textit{hole} \\
    & | & C \: M & \textit{left application} \\
    & | & V \: C & \textit{right application} \\
    & | & \mathrm{let} \; x = C \; \mathrm{in} \; M &
    \textit{let-context} \\ 
    \mathrm{\mathbf{Evaluation}} & & & \boxed{M\leadsto N} \\
    \mathrm{c} \: V & \leadsto &
    \left\llbracket\mathrm{c}\right\rrbracket(V) &
    [\mathcal{E}-Constant] \\ 
    (\lambda x.M) V & \leadsto & [V/x]M & [\mathcal{E}-\beta] \\
    \mathrm{let} \; x = V \; \mathrm{in} \; M & \leadsto & [V/x]M &
    [\mathcal{E}-Let] \\
    C[M] & \leadsto & C[N] \; if \; M\leadsto N & [\mathcal{E}-Compat]
  \end{array}
  \]
  \caption{Small-step operational semantics.}
  \label{fig:op}
\end{figure}



\subsection{Typing rules}
\label{sec:ty-rules}

We present our typing rules via the collection of derivation rules
shown in Figure \ref{fig:type-system}. We present three different
judgments: \textbf{type judgment}, of the form
$\judgment{M}{\sigma}$ meaning that term $M$ has
type $\sigma$ under environment
$\Gamma$, 
restricted to the qualifiers contained in $\mathbb{Q}$, i.e., only
expressions from the set $\mathbb{Q}$ can be used as refinement
predicates for the following terms: let bindings,
$\lambda$-abstractions and type instantiations;
\textbf{subtype judgment} $\Gamma\vdash^\cap\sigma_1
\prec   
\sigma_2$, stating that $\sigma_1$
is a subtype of $\sigma_2$ under 
the conditions of environment $\Gamma$; and the
\textbf{well-formedness judgment}
$\Gamma\vdash^\cap\sigma$ indicating 
that variables referred by the refinements of
$\sigma$ are in the 
scope of corresponding expressions. The well-formedness judgment can
be lifted to well-formedness of environments, by stating that an
environment is well-formed if for every binding, types are well-formed
with respect to the prefix environment.
This well-formedness restriction
implies the absence of the structural property of exchange in our
system, since by permuting the bindings in $\Gamma$ one could generate
an inconsistent environment 

\begin{figure}
\small
\begin{types}[b]{\textwidth}
\textbf{Liquid Intersection Type checking}
\hfill
$\boxed{\Gamma\vdash^\cap_{\mathbb{Q}} M :  {\sigma}}$
\[
\begin{array}{c}
\inferrule*[Lab={\scriptsize Sub}]
{
  \Gamma\vdash^\cap_{\mathbb{Q}} M: {\sigma}_1 \\
  \Gamma\vdash^\cap {\sigma}_1\prec {\sigma}_2
  \\  
  \Gamma\vdash^\cap {\sigma}_2
}
{
  \Gamma\vdash^\cap_{\mathbb{Q}} M:\sigma_2
} \qquad

\inferrule*[Lab=\scriptsize Intersect]
{
  \Gamma\vdash^\cap_{\mathbb{Q}} M:{\tau}_1 \\
  \Gamma\vdash^\cap_{\mathbb{Q}} M:{\tau}_2 \\
  {\tau}_1\cap{\tau}_2 :: T
}
{
  \Gamma\vdash^\cap_{\mathbb{Q}}
  M:{\tau}_1\cap{\tau}_2 
} \\[10pt]

\inferrule*[Lab=\scriptsize Var-B]
{
  \Gamma(x) =
  {\tau}_1\cap\ldots\cap{\tau}_n \\ 
  {\tau}_i :: B \,(\forall i :  1 \leq i \leq n)
}
{
  \Gamma\vdash^\cap_{\mathbb{Q}} x:\{v: B | v = x\}
} \qquad

\inferrule*[Lab=\scriptsize Var]
{
  \Gamma(x) \: \textrm{not a base type} \\ \Gamma(x) :: T
}
{
  \Gamma \vdash^\cap_{\mathbb{Q}} x : \Gamma(x)
} \\[10pt]

\inferrule*[Lab=\scriptsize App]
{
  \Gamma\vdash^\cap_{\mathbb{Q}}
  M:(x:{\tau}_x\rightarrow {\tau}) \\
  \Gamma\vdash^\cap_{\mathbb{Q}} N:{\tau}_x 
}
{
  \Gamma\vdash^\cap_{\mathbb{Q}} MN:[N/x]{\tau}
} \qquad

\inferrule*[Lab=\scriptsize Fun]
{
  \Gamma;x:\hat{\tau}_x\vdash^\cap_{\mathbb{Q}} M:\hat{\tau} \\
  \Gamma\vdash^\cap\hat{\tau_x}\rightarrow\hat{\tau} \\
  \hat{\tau} :: T
}
{
  \Gamma\vdash^\cap_{\mathbb{Q}}\lambda x.M :
  (x : \hat{\tau}_x\rightarrow\hat{\tau})
} \\[10pt]

\inferrule*[Lab=\scriptsize Const]
{
}
{
  \Gamma\vdash^\cap_{\mathbb{Q}}\mathrm{c} : ty(\mathrm{c})
} \qquad

\inferrule*[Lab=\scriptsize Let]
{
  \Gamma\vdash^\cap_{\mathbb{Q}} M : {\sigma} \\
  \Gamma;x:{\sigma}\vdash^\cap_{\mathbb{Q}} N :
  \hat{\tau} \\
  \Gamma\vdash^\cap \hat{\tau}
}
{
  \Gamma\vdash^\cap_{\mathbb{Q}} \mathrm{let} \; x = M \; \mathrm{in}
  \; N : \hat{\tau} 
} \\[10pt]

\inferrule*[Lab=\scriptsize Gen]
{
  \judgment{M}{\sigma} \\
  \alpha\not\in\Gamma
}
{
  \judgment{[\Lambda\alpha]M}{\forall\alpha.{\sigma}}
} \qquad

\inferrule*[Lab=\scriptsize Inst]
{
  \judgment{M}{\forall\alpha.{\sigma}} \\
  \Gamma\vdash^\cap \hat{\tau}\\
  \textrm{Shape}(\hat{\tau}) = T
}
{
  \judgment{[T]M}{[\hat{\tau}/\alpha]{\sigma}}
}
\end{array}
\]

\textbf{Subtyping} \hfill $\boxed{\Gamma\vdash^\cap
  {\sigma_1} \prec {\sigma_2}}$
\[
\begin{array}{c}
  \inferrule*[Lab=\scriptsize $\prec$-Base]
  {
    \mathtt{Valid}(\left\llbracket\Gamma
    \right\rrbracket\wedge(\left\llbracket  
      E_1 
    \right\rrbracket\wedge\ldots\wedge\left\llbracket
      E_n\right\rrbracket)\Rightarrow(\left\llbracket
      E_{1}'\right\rrbracket\wedge\ldots\wedge\left\llbracket
      E_{m}'\right\rrbracket))      
  }
  {
    \Gamma\vdash^\cap \left\{v : B \: | \:
      E_1\right\}\cap\ldots\cap\left\{v : B \: | \: E_n\right\}
    \prec \left\{v : B \: | \:
      E_{1}'\right\}\cap\ldots\cap\left\{v : B \: | \: E_{m}'\right\}
  } \\[10pt]

  \inferrule*[Lab=\scriptsize $\prec$-Intersect-Fun]
  {
  }
  {
    \Gamma\vdash^\cap(x:{\tau}_x
    \rightarrow{\tau}_1)\cap  
    (x:{\tau}_x\rightarrow{\tau}_2)\prec
    (x:{\tau}_x\rightarrow
    {\tau}_1\cap{\tau}_2)      
  } \qquad 
  \inferrule*[Right=$i\in\{1{,}2\}$, Lab=\scriptsize $\prec$-Elim]
  {
  }
  {
    \Gamma\vdash^\cap
    {\tau}_1\cap
    {\tau}_2\prec{\tau}_i    
  } 
  \\[10pt]
  
  \inferrule*[Lab=\scriptsize $\prec$-Fun]
  {
    \Gamma\vdash^\cap
    {\tau}_x'\prec{\tau}_x \\  
    \Gamma;x:{\tau}_x'\vdash^\cap
    {\tau}\prec{\tau}' 
  }
  {
    \Gamma\vdash^\cap
    x:{\tau}_x\rightarrow{\tau}\prec
    x:{\tau}_x'\rightarrow{\tau}' 
  } \hspace{50pt}
  
  \inferrule*[Lab=\scriptsize $\prec$-Var]
  {
  }
  {
    \Gamma\vdash^\cap \alpha\prec\alpha
  } \\[10pt]

  \inferrule*[Lab=\scriptsize $\prec$-Intersect]
  {
    \Gamma\vdash^\cap {\tau}\prec{\tau}_1 \\
    \Gamma\vdash^\cap {\tau}\prec{\tau}_2
  }
  {
    \Gamma\vdash^\cap
    {\tau}\prec
    {\tau}_1\cap{\tau}_2   
  } \hspace{50pt}

  \inferrule*[Lab=\scriptsize $\prec$-Poly]
  {
    \Gamma\vdash^\cap {\sigma}_1\prec
    {\sigma}_2 
  }
  {
    \Gamma\vdash^\cap\forall\alpha.{\sigma}_1\prec
    \forall\alpha.{\sigma}_2 
  }
\end{array}
\]
\textbf{Well formed types} \hfill $\boxed{\Gamma\vdash^\cap
  {\sigma}}$ 
\[
\begin{array}{c}
  \inferrule*[Lab=\scriptsize WF-B]
  {
    \Gamma;\nu:B \vdash^\cap E : bool
  }
  {
    \Gamma\vdash^\cap \left\{\nu: B \: | \: E\right\}
  } \hspace{50pt}
  
  \inferrule*[Lab={\scriptsize WF-Var}]
  {
  }
  {
    \Gamma\vdash^\cap\alpha
  } \\[10pt]

  \inferrule*[Lab=\scriptsize WF-Fun]
  {
    \Gamma;x: {\tau}_x \vdash^\cap {\tau}
  }
  {
    \Gamma \vdash^\cap
    x:{\tau}_x\rightarrow{\tau} 
  }
  \hspace{50pt}

  \inferrule*[Lab=\scriptsize WF-Poly]
  {
    \Gamma\vdash^\cap{\sigma}
  }
  {
    \Gamma\vdash^\cap\forall\alpha.{\sigma}
  } \\[10pt]
  
  \inferrule*[Lab=\scriptsize WF-Intersect]
  {
    \Gamma\vdash^\cap {\tau}_1 \\ \Gamma\vdash^\cap
    {\tau}_2 
  }
  {
    \Gamma\vdash^\cap {\tau}_1\cap{\tau}_2
  } 
\end{array}
\]
\end{types}
\caption{Typing rules for Liquid Intersection Types.}
\label{fig:type-system}
\end{figure}

The rule [{\sc App}] conforms to the dependent types discipline, since
the type of an application $MN$ is the return type of $M$ but with
every occurrence of $x$ in the refinements substituted by $N$. 

Another point
worth mentioning is the distinction made when the type of a variable
is to be retrieved, rules [{\sc Var-B}] and [{\sc Var}]. Whenever the
type of the variable $z$ is an intersection of refined basic type we
ignore these refinements and assign $z$ the type $\left\{\nu : B \: |
  \: \nu = z\right\}$, for some basic type $B$. This is inspired on
the system of Liquid Types \cite{Rondon:2008:LT:1375581.1375602},
since this assigned refined type is very useful when it comes to use
in subtyping, especially with the rule [$\prec$-Base]. When this is
not the case, the type of a variable is the one stored in $\Gamma$.

One novel aspect of this system is the presence of the [{\sc
  Intersect}] rule, which allows to intersect two types that have been 
derived for the same term. The use of this rule increases the
expressiveness of the types language itself, since more detailed
types can be derived for a program.

The subtyping relation presents some typical rules for a system with
intersection types. These allow to capture the relations at the level
of intersections in types, with no concern for the refinements of the
two types being compared. On the other side, comparing two refined
base types reduces to the check of an implication formula between the
refinement expressions. Our system uses a decidable notion of
implication in the rule [$\prec$-Base], by embedding environments and
refinement expressions into a decidable logic. This logic
contains at least equality, uninterpreted functions and linear
arithmetic. This is the core logical setting of most state-of-the-art
SMT solvers. The embedding $\left\llbracket M \right\rrbracket$
translates the term $M$ to the correspondent one in the
logic (if it is the case $M$ is a constant or an arithmetic
operator), or if $M$ is a $\lambda$-abstraction or an application
encodes it via uninterpreted functions. The embedding of environments
is defined as
\[
\left\llbracket\Gamma\right\rrbracket \triangleq \bigwedge\left\{(
  \left\llbracket E_1 \right\rrbracket\wedge\ldots
  \wedge\left\llbracket E_n \right\rrbracket)[x/\nu] \: | \: x :
  \left\{\nu : B \: | \: E_1 \right\}\cap\ldots 
  \cap\left\{\nu : B \: | \: E_n \right\} \in \Gamma \right\}  
\]

Given that every implication expression generated in rule [$\prec$-Base] is
decidable, it is then suitable to be discharged by some automatic
theorem prover. So, type-checking in our system can be
seen as a \emph{typing-and-proof} process.


We show an example of a derivation for the term $\lambda x.-x$,
assuming $\mathbb{Q} = \{\nu \geq 0, \nu \leq 0\}$. With
$\Gamma = x:\left\{\nu\geq 0\right\}$, consider:
\[
\small
\begin{array}{lc} 
  \mathcal{D}_1': & \\
  & \inferrule*
  {
    \inferrule*[Left=\footnotesize Const]
    {
    }
    {
      \Gamma\vdash^\cap_{\mathbb{Q}} - :
      (y:int\rightarrow\left\{\nu=-y\right\}) 
    } \\
    \inferrule*[Right=\footnotesize Sub]
    { 
      \inferrule*[Left=\footnotesize Var-B]
      {
        \Gamma(x) = \left\{\nu \geq 0\right\}
      }
      {
        \Gamma\vdash^\cap_{\mathbb{Q}} x: \left\{\nu = x
        \right\} 
      } \\
      \inferrule*[Right=\footnotesize$\prec$-Base]
      {
        \mathrm{Valid}(x \geq 0 \wedge \nu = x \Rightarrow \top)
      }
      {
        \Gamma\vdash^\cap \left\{\nu = x
        \right\}\prec int
      }
    }
    {
      \Gamma\vdash^\cap_{\mathbb{Q}} x: int
    }
  }
  {
    \Gamma\vdash^\cap_{\mathbb{Q}} -x :
    \left\{\nu=-x\right\}
  }
\end{array}
\]
and:
\[
\small
\begin{array}{lc} 
  \mathcal{D}_1: & \\
  & \inferrule*[Right=\footnotesize Fun]
  {
    \inferrule*[Right=\footnotesize Sub]
    {
      \inferrule*[]
      {}
      {
        \mathcal{D}_1'
      }\\
      \inferrule*[Right=\footnotesize $\prec$-Base]
      {
        \mathrm{Valid}(x\geq 0 \wedge \nu=-x\Rightarrow \nu\leq 0)
      }
      {
        \Gamma\vdash^\cap\left\{\nu=-x\right\}\prec\left\{\nu\leq 0\right\}
      }
    }
    {
      \Gamma\vdash^\cap_{\mathbb{Q}} -x: \left\{\nu \leq 0\right\}
    }
  }
  {
    \vdash^\cap_{\mathbb{Q}} \lambda x.-x : (x:\left\{\nu \geq
      0\right\}\rightarrow\left\{\nu \leq 0\right\} )
  } 
\end{array}
\]

We can also derive
$\vdash^\cap_{\mathbb{Q}}\lambda x.-x : (x:\left\{\nu \leq
  0\right\}\rightarrow\left\{\nu \geq 0\right\})$ (similarly to the
previous derivation, with the corresponding $\leq$ and $\geq$ symbols
changed). Naming that derivation $\mathcal{D}_2$, we
finally have: 
\[
\small
\inferrule*[Right=\footnotesize{Intersect}]
    {
      \mathcal{D}_1 \\
      \mathcal{D}_2
    } 
    { \vdash^\cap_{\mathbb{Q}} \lambda x.-x : (x:\left\{\nu \geq 0\right\}\rightarrow\left\{\nu \leq 0\right\})\cap(x:\left\{\nu \leq 0\right\}\rightarrow\left\{\nu \geq 0\right\}) 
    }
\]
We omit the well-formedness and well-founded sub-derivations, since
they are trivially constructed and use $int$ to denote the type $\{\nu
: int \: | \: \top\}$, that is, the common type for integer
values. 

\subsection{Properties}
\label{sec:properties}

In order to prove soundness properties for our system we follow the
approach of
\cite{Rondon:2008:LT:1375581.1375602,Vazou:2013:ART:2450268.2450286}.
The decidable notion of implication checking employed by the subtyping
rules is a problem when it comes to prove a substitution lemma. So,
instead we prove subject reduction for a version of the system with
undecidable subtyping and unrestricted expressions in refinement
predicates. The typing judgment in this system will be denoted by
$\Gamma\vdash^\cap M : \sigma$, and the inference rules are presented
in Figures~\ref{fig:dep-type-system} and~\ref{fig:wf-cs-rules}. Then,
we show that any derivation in the decidable system has a counter-part
in the undecidable one. We present in this section the more
interesting steps employed during the proof of subject reduction for
our type system. The detailed proofs can be found in~\cite{Pereira14}. 





\begin{figure}
\begin{types}[b]{\textwidth}
\small
\textbf{Refinement Intersection type checking}
\hfill
$\boxed{\Gamma\vdash^\cap M : \sigma}$
\[
\begin{array}{cc}
\inferrule*[Lab=\scriptsize Sub]
{
  \Gamma\vdash^\cap M:\sigma_1 \\
  \Gamma\vdash^\cap\sigma_1\prec\sigma_2 \\ 
  \Gamma\vdash^\cap \sigma_2
}
{
  \Gamma\vdash^\cap M:\sigma_2
} &

\inferrule*[Lab=\scriptsize Intersect]
{
  \Gamma\vdash^\cap M:\tau_1 \\
  \Gamma\vdash^\cap M:\tau_2 \\
  \tau_1\cap\tau_2 :: T
}
{
  \Gamma\vdash^\cap M:\tau_1\cap\tau_2
} \\[10pt]

\inferrule*[Lab=\scriptsize Var-B]
{
  \Gamma(x) = \tau_1\cap\ldots\cap\tau_n \\ \tau_i :: B \,
  (\forall i :  1 \leq i \leq n)
}
{
  \Gamma\vdash^\cap x:\{\nu: B \: | \: \nu = x\}
} &

\inferrule*[Lab=\scriptsize Var]
{
  \Gamma(x) \: \textrm{not a base type} \\ \Gamma(x) :: T
}
{
  \Gamma \vdash^\cap x : \Gamma(x)
} \\[10pt]

\inferrule*[Lab=\scriptsize App]
{
  \Gamma\vdash^\cap M:(x:\tau_x\rightarrow\tau) \\
  \Gamma\vdash^\cap N:\tau_x
}
{
  \Gamma\vdash^\cap MN:[N/x]\tau
} &

\inferrule*[Lab=\scriptsize Fun]
{
  \Gamma;x:\tau_x\vdash^\cap M:\tau \\
  \Gamma\vdash^\cap\tau_x\rightarrow\tau \\
  \tau :: T
}
{
  \Gamma\vdash^\cap\lambda x.M : (x:\tau_x\rightarrow\tau) 
} \\[10pt]

\inferrule*[Lab=\scriptsize Const]
{
}
{
  \Gamma\vdash^\cap\mathrm{c} : ty(\mathrm{c})
} &

\inferrule*[Lab=\scriptsize Let]
{
  \Gamma\vdash^\cap M : \sigma \\
  \Gamma;x:\sigma\vdash^\cap N : \tau \\
  \Gamma\vdash^\cap \tau
}
{
  \Gamma\vdash^\cap \mathrm{let} \; x = M \; \mathrm{in} \; N : \tau
} \\[10pt]

\inferrule*[Lab=\scriptsize Gen]
{
  \Gamma\vdash^\cap M : \sigma \\
  \alpha\not\in\Gamma
}
{
  \Gamma\vdash^\cap [\Lambda\alpha]M : \forall\alpha.\sigma
} &

\inferrule*[Lab=\scriptsize Inst]
{
  \Gamma\vdash^\cap M : \forall\alpha.\sigma \\
  \Gamma\vdash^\cap\tau \\
  \textrm{Shape}(\tau) = T
}
{
  \Gamma\vdash^\cap [T]M : [\tau/\alpha]\sigma
}
\end{array}
\]

\textbf{Implication} \hfill $\boxed{\Gamma\vdash^\cap E \Rightarrow E'}$
\[
\inferrule*[Lab=\scriptsize Imp]
{
  \Gamma\vdash^\cap E : bool \\ 
  \Gamma\vdash^\cap E': bool \\
  \forall\rho.(\Gamma\models\rho\; \textrm{and}\;
  \rho (E)\overset{*}{\leadsto}\top\; \textrm{implies}\;
  \rho (E')\overset{*}{\leadsto}\top) 
}
{
  \Gamma\vdash^\cap E\Rightarrow E'
}
\]

\textbf{Subtyping} \hfill $\boxed{\Gamma\vdash^\cap \sigma_1 \prec
  \sigma_2}$
\[
\begin{array}{c}
  \inferrule*[Lab=\scriptsize $\prec$-Base]
  {
    \Gamma; \nu:B \vdash^\cap
    E_1\wedge\ldots\wedge E_n\Rightarrow
    E_1'\wedge\ldots\wedge E_m' 
  }
  {
    \Gamma\vdash^\cap \left\{\nu : B \: | \:
      E_1\right\}\cap\ldots\cap\left\{\nu : B \: | \: E_n\right\}
    \prec \left\{\nu : B \: | \:
      E_{1}'\right\}\cap\ldots\cap\left\{\nu : B \: | \: E_{m}'\right\}
  } \\[10pt]

  \inferrule*[Lab=\scriptsize $\prec$-Intersect-Fun]
  {
  }
  {
    \Gamma\vdash^\cap(x:\tau_x\rightarrow\tau_1)
    \cap(x:\tau_x\rightarrow\tau_2)\prec  
    (x:\tau_x\rightarrow\tau_1\cap\tau_2)   
  } \qquad 

  \inferrule*[Right=$i\in\{1{,}2\}$, Lab=\scriptsize $\prec$-Elim]
  {
  }
  {
    \Gamma\vdash^\cap \tau_1\cap\tau_2\prec\tau_i 
  } \\[10pt]
  
  \inferrule*[Lab=\scriptsize $\prec$-Fun]
  {
    \Gamma\vdash^\cap \tau_x'\prec\tau_x \\ 
    \Gamma;x:\tau_x'\vdash^\cap \tau\prec\tau'
  }
  {
    \Gamma\vdash^\cap (x:\tau_x\rightarrow\tau)\prec
    (x:\tau_x'\rightarrow\tau') 
  } \hspace{50pt}
  
  \inferrule*[Lab=\scriptsize $\prec$-Var]
  {
  }
  {
    \Gamma\vdash^\cap \alpha\prec\alpha
  } \\[10pt]

  \inferrule*[Lab=\scriptsize $\prec$-Intersect]
  {
    \Gamma\vdash^\cap \tau\prec\tau_1 \\
    \Gamma\vdash^\cap \tau\prec\tau_2
  }
  {
    \Gamma\vdash^\cap \tau\prec\tau_1\cap\tau_2
  } \hspace{50pt}

  \inferrule*[Lab=\scriptsize $\prec$-Poly]
  {
    \Gamma\vdash^\cap \sigma_1\prec\sigma_2
  }
  {
    \Gamma\vdash^\cap\forall\alpha.\sigma_1\prec\forall\alpha.\sigma_2
  }
\end{array}
\]
\end{types}
\caption{Refinement Intersection typing rules}
\label{fig:dep-type-system}
\end{figure}

\begin{figure}
\begin{types}[b]{\textwidth}
    \small
    \textbf{Well formed types} \hfill $\boxed{\Gamma\vdash^\cap \sigma}$
\[
\begin{array}{c}
  \inferrule*[Lab=\scriptsize WF-B]
  {
    \Gamma;\nu:B \vdash^\cap \phi : bool
  }
  {
    \Gamma\vdash^\cap \left\{\nu: B \: | \: \phi\right\}
  } \hspace{75pt}
  
  \inferrule*[Lab=\scriptsize WF-Var]
  {
  }
  {
    \Gamma\vdash^\cap\alpha
  } \\[10pt]

  \inferrule*[Lab=\scriptsize WF-Fun]
  {
    \Gamma;x: \tau_x \vdash^\cap \tau
  }
  {
    \Gamma \vdash^\cap (x:\tau_x\rightarrow\tau)
  }
  \hspace{75pt}

  \inferrule*[Lab=\scriptsize WF-Poly]
  {
    \Gamma\vdash^\cap\sigma
  }
  {
    \Gamma\vdash^\cap\forall\alpha.\sigma
  } \\[10pt]
  
  \inferrule*[Lab=\scriptsize WF-Intersect]
  {
    \Gamma\vdash^\cap \tau_1 \\ \Gamma\vdash^\cap \tau_2
  }
  {
    \Gamma\vdash^\cap \tau_1\cap\tau_2
  } 
\end{array}
\]

\textbf{Consistent substitutions} \hfill $\boxed{\Gamma\models\rho}$
\[
\begin{array}{lr}
  \inferrule*[Lab=\scriptsize CS-Empty]
  {
  }
  {
    \emptyset \models \emptyset
  } \quad & \quad
  \inferrule*[Lab=\scriptsize CS-Ext]
  {
    \Gamma\models\rho \\
    \emptyset\vdash^\cap V : \rho(\sigma)
  }
  {
    \Gamma;x:\sigma \models \rho; [V/x]
  }
\end{array}
\]
\end{types}
\caption[Well formed Types and consistent substitutions.]{Rules for
  well formed Refinement Intersection Types and 
  consistent substitutions.}
\label{fig:wf-cs-rules}
\end{figure}

\begin{mydef}[Constants]
  Each constant \textrm{c} has a type $ty(\textrm{c})$ such that:
  \begin{enumerate}
  \item $\emptyset\vdash^\cap ty(\textrm{c})$;
  \item if \textrm{c} is a primitive function then it cannot get
    stuck, thus if $\Gamma\vdash^\cap\mathrm{c}\: v$ then
    $\llbracket\mathrm{c}\rrbracket(v)$ is defined and if
    $\Gamma\vdash^\cap \mathrm{c} \:M : \sigma$ and
    $\llbracket\mathrm{c}\rrbracket(M)$ is defined then
    $\Gamma\vdash^\cap \llbracket\mathrm{c}\rrbracket\:(M) : \sigma$;
  \item if $ty(\textrm{c})$ is $\left\{\nu:B \: | \: \phi\right\}$
    then $\phi\equiv\nu = c$.
  \end{enumerate}
\label{def:constants}
\end{mydef}

\begin{mydef}[Embedding]
  The embedding $\left\llbracket\cdot\right\rrbracket$ is defined as a
  map from terms and environments to formulas in the decidable logic
  such that for all $\Gamma, E, E'$ if $\Gamma\vdash^\cap
  E:bool$, $\Gamma\vdash^\cap E':bool$,
  $\mathrm{Valid}(\llbracket\Gamma\rrbracket
  \wedge\llbracket E \rrbracket\Rightarrow E')$,    
  then $\Gamma\vdash^\cap E \Rightarrow E'$.
\end{mydef}

\begin{mydef}[Substitution]
  We define substitution on types, $\rho(\sigma)$, as follows:
  \[
  \begin{array}{rcl}
    \rho(\alpha) & = & \alpha \\
    \rho(\left\{\nu : B \: | \: E\right\}) & = & \left\{\nu : B \: | \:
      \rho (E)\right\} \\
    \rho(x:{\tau}_x\rightarrow{\tau}) & = &
    x:\rho({\tau}_x)\rightarrow \rho({\tau}) \\
    \rho(\forall\alpha.{\sigma}) & = & \forall\alpha.\rho({\sigma}) \\
    \rho({\tau}_1\cap{\tau}_2) & = &
    \rho({\tau}_1) \cap \rho({\tau}_2) 
  \end{array}
  \]
\end{mydef}

A substitution can be lifted to typing contexts as expected:
\[
\begin{array}{rcl}
  \rho(\emptyset) & = & \emptyset \\
  \rho(\Gamma;x : \sigma) & = & \rho(\Gamma);x : \rho(\sigma)
\end{array}
\]

\begin{mydef}[Domain of a substitution]
  The domain of a substitution, $\mathrm{\mathbf{Dom}}(\rho)$, is
  defined as follows:
  \[
  \begin{array}{rcl}
    \mathrm{\mathbf{Dom}}(\emptyset) & = & \{\} \\
    \mathrm{\mathbf{Dom}}(\rho;[V/x]) & = &
    \mathrm{\mathbf{Dom}}(\rho) \cup \{x\}
  \end{array}
  \]
\end{mydef}

\begin{lemma}[Substitution permutation]
  If $\Gamma\models\rho_1;\rho_2$ then
  \begin{enumerate}
  \item $\mathrm{\mathbf{Dom}}(\rho_1)\cap
    \mathrm{\mathbf{Dom}}(\rho_2) = \emptyset$;
  \item for all Liquid Intersection Type $\sigma$,
    $\rho_1;\rho_2(\sigma) = \rho_2;\rho_1(\sigma)$.
  \end{enumerate}
\label{lemma:permutation}
\end{lemma}

\begin{proof}
  \begin{enumerate}
  \item By induction on the derivation $\Gamma\models\rho_1;\rho_2$,
    splitting cases on which rule was used at the bottom.
  \item By induction on the structure of $\sigma$. \qed
  \end{enumerate}
\end{proof}

\begin{lemma}[Well-formed substitutions] \hfill
  \begin{enumerate}
  \item If $\Gamma\models\rho_1;\rho_2$ then there are
    $\Gamma_1,\Gamma_2$ such that $\Gamma=\Gamma_1;\Gamma_2$,
    $\mathrm{\mathbf{Dom}}(\rho_1) = \mathrm{\mathbf{Dom}}(\Gamma_1)$,
    $\mathrm{\mathbf{Dom}}(\rho_2) = \mathrm{\mathbf{Dom}}(\Gamma_2)$;
  \item $\Gamma_1;\Gamma_2\models\rho_1;\rho_2,
    \mathrm{\mathbf{Dom}}(\rho_1) = \mathrm{\mathbf{Dom}}(\Gamma_1),
    \mathrm{\mathbf{Dom}}(\rho_2) = \mathrm{\mathbf{Dom}}(\Gamma_2)$
    iff $\Gamma_1\models\rho_1$, $\rho_1\Gamma_2\models\rho_2$.
  \end{enumerate}
\label{lemma:wf-subs}
\end{lemma}

\begin{proof}
  \begin{enumerate}
  \item By induction on the structure of $\Gamma$.
  \item By induction on the structure of $\Gamma_2$. \qed
  \end{enumerate}
\end{proof}

\begin{corollary}[Well-formed substitutions] \hfill

  $\Gamma_1;x : \sigma_x; \Gamma_2\models \rho_1;[V_x/x];\rho_2 
  \Longleftrightarrow \Gamma_1\models\rho_1$, \; $\emptyset\vdash V_x :
  \rho_1(\sigma_x)$, \; $\rho_1;[V_x/x](\Gamma_2)\models \rho_2$.
\label{cor}
\end{corollary}

\begin{proof} Corollary of Lemma~\ref{lemma:wf-subs}. \qed
\end{proof}

\begin{lemma}[Weakening]
  Let
  \[
  \begin{array}{l}
    \Gamma=\Gamma_1;\Gamma_2 \\
    \Gamma'=\Gamma_1;x:\sigma_x;\Gamma_2 \\
    x \not\in \mathrm{FV}(\Gamma_2)
  \end{array}
  \]
  then:
  \begin{enumerate}
  \item if $\Gamma'\models\rho_1;[V/x];\rho_2$, then
    $\Gamma\models\rho_1;\rho_2$;
  \item if $\Gamma\vdash^\cap E\Rightarrow E'$, then
    $\Gamma'\vdash^\cap E\Rightarrow E'$;
  \item if $\Gamma\vdash^\cap\sigma_1\prec\sigma_2$, then
    $\Gamma'\vdash^\cap\sigma_1\prec\sigma_2$;
  \item if $\Gamma\vdash^\cap\sigma$, then $\Gamma'\vdash^\cap\sigma$;
  \item if $\Gamma\vdash^\cap M : \sigma$, then $\Gamma'\vdash^\cap M :
    \sigma$.
  \end{enumerate}
\label{lemma:weakening}
\end{lemma}

\begin{proof}
  By simultaneous induction on the derivations of the antecedent
  judgments. \qed
\end{proof}

\begin{lemma}[Substitution] 
  If 
  \[
  \begin{array}{l}
    \Gamma_1\vdash^\cap V : \sigma' \\
    \Gamma=\Gamma_1;x:\sigma';\Gamma_2 \\
    \Gamma'=\Gamma_1;[V/x]\Gamma_2
  \end{array}
  \]
  then:
  \begin{enumerate}
  \item if $\Gamma\models\rho_1;[V/x]\rho_2$, then
    $\Gamma'\models\rho_1;\rho_2$;
  \item if $\Gamma\vdash^\cap E\Rightarrow E'$, then
    $\Gamma'\vdash^\cap[V/x]E\Rightarrow[V/x]E'$;
  \item if $\Gamma\vdash^\cap\sigma_1\prec\sigma_2$, then
    $\Gamma'\vdash^\cap[V/x]\sigma_1\prec[V/x]\sigma_2$;
  \item if $\Gamma\vdash^\cap\sigma$, then
    $\Gamma'\vdash^\cap[V/x]\sigma$;
  \item if $\sigma :: T$, then $[V/x]\sigma :: T$;
  \item if $\Gamma\vdash^\cap M : \sigma$, then $\Gamma'\vdash^\cap M :
    \sigma$.
  \end{enumerate}
\label{lemma:substitution}
\end{lemma}

\begin{proof}
  By simultaneous induction on the derivations of the antecedent judgments. \qed
\end{proof}

\begin{theorem}[Subject reduction] If $\Gamma\vdash^\cap M : \sigma$
  and $M\leadsto N$, then $\Gamma\vdash^\cap N : \sigma$.
\label{th:subject-reduction}
\end{theorem}

\begin{proof} By induction on the derivation $\Gamma\vdash^\cap M :
  \sigma$, splitting cases on which rule was used at the bottom. We
  give here the cases for [{\scshape Intersect}] and [{\scshape App}].
  \begin{itemize}
  \item case [{\scshape Intersect}]: By inversion
    \[
    \begin{array}{l}
      \Gamma\vdash^\cap M : \tau_1 \\
      \Gamma\vdash^\cap M : \tau_2 \\
      \tau_1\cap\tau_2 :: T
    \end{array}
    \]

    By IH
    \[
    \begin{array}{l}
      \Gamma\vdash^\cap N : \tau_1 \\
      \Gamma\vdash^\cap N : \tau_2
    \end{array}
    \]

    So, the following derivation is then valid
    \[
    \inferrule*[Left=Intersect]
    {
      \Gamma\vdash^\cap N : \tau_1 \\
      \Gamma\vdash^\cap N : \tau_2 \\
      \tau_1\cap\tau_2 :: T
    }
    {
      \Gamma\vdash^\cap N : \tau_1\cap\tau_2
    }
    \]

  \item case [{\scshape App}]: By inversion
    \[
    \begin{array}{l}
      \Gamma\vdash^\cap M : (x:\tau_x\rightarrow\tau) \\
      \Gamma\vdash^\cap N : \tau_x
    \end{array}
    \]

    \begin{itemize}
    \item sub-case in which $M$ is a context: For this case consider
      $M\leadsto M'$.

      By IH
      \[
      \Gamma\vdash^\cap M': (x:\tau_x\rightarrow\tau)
      \]

      Given that $M\leadsto M'$, then $MN\leadsto M'N$.

      The following derivation is then valid
      \[
      \inferrule*[Left=App]
      {
        \Gamma\vdash^\cap M': (x:\tau_x\rightarrow\tau) \\
        \Gamma\vdash^\cap N : \tau_x
      }
      {
        \Gamma\vdash^\cap M'N : [N/x]\tau
      }
      \]

    \item sub-case in which $N$ is a context: Similar to the previous
      one.

    \item sub-case in which application is of the form
      $\mathrm{c}\,V$: By pushing applications of rule [{\scshape
        Sub}] down, we can ensure rule [{\scshape Const}] was used at
      the bottom of the derivation of the type for $\mathrm{c}$. 

      For this case, $\mathrm{c}\, V\leadsto
      \llbracket\mathrm{c}\rrbracket(V)$.

      By inversion
      \[
      \begin{array}{l}
        \Gamma\vdash^\cap \mathrm{c} : (x : \tau_x\rightarrow\tau)
        \\
        \Gamma\vdash^\cap V : \tau_x
      \end{array}
      \]

      By Definition~\ref{def:constants}, we have
      \[
      \Gamma\vdash^\cap \llbracket\mathrm{c}\rrbracket(V) : [V/x]\tau 
      \]
      which is the desired conclusion.

    \item case in which application is of the form $(\lambda x.M) V$:
      For this case 
      \[
      (\lambda x.M) V \leadsto [V/x]M
      \]

      By pushing applications of the rule [{\scshape Sub}] down, we
      can ensure rule [{\scshape Fun}] is used at the bottom of the
      derivation of the type for $\lambda x.M$.

      By inversion
      \[
      \begin{array}{l}
        \Gamma\vdash^\cap \lambda x.M : (x:\tau_x\rightarrow \tau)
        \\
        \Gamma\vdash^\cap V : \tau_x
      \end{array}
      \]

      By inversion on rule [{\scshape Fun}]
      \[
      x : \tau_x \vdash^\cap M : \tau
      \]

      By Lemma~\ref{lemma:substitution}
      \[
      \Gamma\vdash^\cap [V/x]M : [V/x]\tau
      \]
      which is the desired conclusion. 
    \end{itemize}
  \end{itemize}
  \qed
\end{proof}

\begin{theorem}[Over approximation] If $\judgment{M}{\sigma}$, then
  $\Gamma\vdash^\cap M : \sigma$.
\label{th:overapproximation}
\end{theorem}

\begin{proof}
The proof follows by straightforward induction on the typing
derivation. At each case the key observation is that each Liquid
Intersection Type is also a Dependent Intersection Type and for each
rule in the decidable system there is a matching rule in the
undecidable side. For the case of [{\scshape $\prec$-Base}] we use
Definition 1. 
\end{proof}

Combining Theorems~\ref{th:subject-reduction}
and~\ref{th:overapproximation} guarantees that at run-time, for 
every well-typed term, taking an evaluation step preserves types.




\section{Type inference}
\label{sec:inference}

In this section we present our algorithm\footnote{For some cases of
  the algorithm we use a \emph{temporary type}, denoted by
  $\mathcal{A}$. The only purpose of temporary types is to ease the
  notation as we explain in section \ref{sec:cons-validity}.} for
inferring Liquid Intersection Types, Figure
\ref{fig:infer-algo}. Before executing this algorithm we bind every
sub expression using the \texttt{let-in} constructor.  This
transformation is closely related with \emph{A}-Normal Forms
\cite{Flanagan:1993:ECC:155090.155113} and is performed to force types
of intermediate expressions to be pushed into the typing context. The
algorithm we propose is built upon three main phases: \textbf{(i)} we
use the ML inference engine to get appropriate types, serving as
\emph{type shapes} for Liquid Intersection Types; \textbf{(ii)} for
some particular sub-terms a set of constraints is generated, ensuring
the well-formedness of types and that subtyping relations hold, in
order to infer sound types; \textbf{(iii)} taking qualifiers from
$\mathbb{Q}$ we solve the generated constraints \textit{on-the-fly},
much like as in classical inference algorithms.


\begin{figure}[h!]
\small
  \[
  \begin{array}{lcl}
    \mathtt{Infer}(\Gamma, x, \mathbb{Q}) & = & \mathrm{if} \;
    \mathcal{W}(\mathtt{Shape}(\Gamma), x) = B\; \mathrm{then} \;
    \left\{v : B \: | \: v = x\right\} \\ 
    & & \mathrm{else} \; \Gamma(x) \\

    \mathtt{Infer}(\Gamma, c, \mathbb{Q}) & = & ty(c) \\

    \texttt{Infer}(\Gamma, \lambda x.M, \mathbb{Q}) & = & \textrm{let}
    \;
    (x:\hat{\tau}_1 \rightarrow \hat{\tau}_1')
    \cap
    \ldots\cap(x:\hat{\tau}_n\rightarrow\hat{\tau}_n')
    = \texttt{Fresh}(\mathcal{W}(\texttt{Shape}(\Gamma), \lambda x.M),
    \mathbb{Q}) \; \textrm{in} \\ 
    & & \textrm{let} \; \tau_i'' = \texttt{Infer}(\Gamma; x:
    \hat{\tau}_i, M, \mathbb{Q}) \; \textrm{in} \\ 
    & & \textrm{let} \; \mathcal{A} =
    \bigcap\left\{(x:\hat{\tau}_j\rightarrow\hat{\tau}_j') \: | \:
      \Gamma\vdash^{\cap}(x:\hat{\tau}_1\rightarrow\hat{\tau}_1')\cap\ldots\cap(x:\hat{\tau}_n\rightarrow\hat{\tau}_n')\right\}
    \; \textrm{in} \\ 
    & & \bigcap\left\{(x:\hat{\tau}_k\rightarrow\hat{\tau}_k') \: | \:
      x:\hat{\tau}_k\rightarrow\hat{\tau}_k' \in \mathcal{A},
      \Gamma;x:\hat{\tau}_k\vdash^{\cap}_{\mathbb{Q}}\tau_k''\prec\hat{\tau}_k'\right\}
    \\ 
    \texttt{Infer}(\Gamma, MN, \mathbb{Q}) & = & \textrm{let} \;
    (x:\tau_1\rightarrow\tau_1')\cap\ldots\cap(x:\tau_n\rightarrow\tau_n')
    = \texttt{Infer}(\Gamma, M, \mathbb{Q}) \; \textrm{in}\\ 
    & & \textrm{let} \; \tau = \texttt{Infer}(\Gamma, N, \mathbb{Q})
    \; \textrm{in}\\ 
    & & \bigcap[N/x]\left\{\tau_i' \: | \:
      \Gamma\vdash^{\cap}_{\mathbb{Q}}\tau\prec\tau_i\right\}\\ 

    \texttt{Infer}(\Gamma, \textrm{let} \; x = M \; \textrm{in} \; N,
    \mathbb{Q}) & = & \textrm{let} \; \hat{\tau} =
    \texttt{Fresh}(\mathcal{W}(\texttt{Shape}(\Gamma), \textrm{let}\;
    x = M \;\textrm{in}\;N), \mathbb{Q}) \; \textrm{in}\\ 
    & & \textrm{let} \; \tau_1 = \texttt{Infer}(\Gamma, M,
    \mathbb{Q}) \; \textrm{in} \\ 
    & & \textrm{let} \; \tau_2 = \texttt{Infer}(\Gamma; x :
    \tau_1, N, \mathbb{Q}) \; \textrm{in}\\ 
    & & \textrm{let} \; \mathcal{A} = \bigcap\left\{\hat{\tau}_i \: | \:
      \Gamma\vdash^{\cap}\hat{\tau}\right\} \; \textrm{in}\\ 
    & & \bigcap\left\{\hat{\tau}_j \: | \: \hat{\tau}_j\in\mathcal{A},
      \Gamma; 
      x:\tau_1\vdash^{\cap}_{\mathbb{Q}}\tau_2\prec\hat{\tau}_j\right\}\\ 

    \texttt{Infer}(\Gamma, [\Lambda\alpha]M, \mathbb{Q}) & = &
    \textrm{let} \; \sigma = \texttt{Infer}(\Gamma, M, \mathbb{Q}) \;
    \textrm{in}\\ 
    & & \forall\alpha.\sigma \\
    
    \texttt{Infer}(\Gamma, [T]M, \mathbb{Q}) & = & \textrm{let} \;
    \tau' = \texttt{Fresh}(T, \mathbb{Q}) \; \textrm{in}\\ 
    & & \textrm{let} \; \forall\alpha.\sigma = \texttt{Infer}(\Gamma,
    M, \mathbb{Q}) \; \textrm{in}\\ 
    & & \texttt{let} \; \mathcal{A} = \bigcap\left\{\tau_i' \: | \:
      \Gamma\vdash^{\cap}\tau'\right\} \; \textrm{in}\\ 
    & & \sigma[\mathcal{A}/\alpha]
  \end{array}
  \]
\caption{Type inference algorithm}
\label{fig:infer-algo}
\end{figure}
\normalsize
\subsection{Using Damas-Milner type inference}
\label{sec:ml-infer}


One key aspect of our inference algorithm is the use of the
inference algorithm $\mathcal{W}$ \cite{Damas:1982:PTF:582153.582176}
to infer ML types. Given the fact that a Liquid Intersection Type for
a term is a refinement and intersections of the corresponding ML type,
the types inferred by $\mathcal{W}$ act as \textit{shapes} for our
Liquid Intersection Types. Indeed, the function
\texttt{Shape}$(\cdot)$ (figuring in the typing rules and in the
inference algorithm) maps a Liquid Intersection Type to its corresponding
ML type. For example, \texttt{Shape}$((x:\left\{\nu = 0
\right\}\rightarrow\left\{\nu = 0\right\})\cap (x:\left\{\nu \geq
  0\right\}\rightarrow\left\{\nu\geq 0\right\})) = int\rightarrow
int$. 

In the inference algorithm, whenever $\mathcal{W}$ is called, we need
to feed it with an environment containing exclusively ML types. This
is done by lifting \texttt{Shape}$(\cdot)$ to 
environments, \texttt{Shape}$(\Gamma)$, by applying it to every
binding in $\Gamma$.

The function $\mathtt{Fresh}(\cdot,\cdot)$ takes an ML type and the
set $\mathbb{Q}$ as input and generates a new Liquid Intersection Type
that contains all the combinations of refinement expressions from
$\mathbb{Q}$. Taking for instance the ML type $T=x:int\rightarrow int$
(we assume we can annotate types with the corresponding abstraction
variable, so it is easier to use with refinements) and $\mathbb{Q} =
\left\{\nu \geq 0, \nu \leq 0\right\}$,
\texttt{Fresh$(T,\mathbb{Q})$} would generate the Liquid
Intersection Type  
\[
\begin{array}{l}
  (x:\left\{\nu \geq 0\right\}\rightarrow\left\{\nu \geq 0
  \right\}) \: \cap \\
  (x:\left\{\nu \geq 0\right\}\rightarrow\left\{\nu \leq 0
  \right\}) \: \cap \\
  (x:\left\{\nu \leq 0\right\}\rightarrow\left\{\nu \geq 0
  \right\}) \: \cap \\
  (x:\left\{\nu \leq 0\right\}\rightarrow\left\{\nu \leq 0 \right\})
\end{array}
\]

\subsection{Constraint generation}
\label{sec:constraint-gen}

The constraints generated during inference serve as a means to ensure 
that the subtyping and well-formedness requirements are respected. In 
the presentation of the algorithm we borrow the notations from the 
typing rules, with $\Gamma\vdash^\cap\sigma$ standing for a 
well-formedness restriction over $\sigma$ and $\Gamma\vdash^\cap
\sigma\prec\sigma'$ constraining type $\sigma$ to be a subtype of
$\sigma'$. 

The well-formedness constraints are generated for terms where a fresh
Liquid Intersection Type is generated ($\lambda$-abstractions,
let-bindings and type application). For a fresh generated Liquid
Intersection Type, solving this kind of constraints will result in a
type where the free variables of every refinement are in scope of the
corresponding expression. 

The second class of constraints are the subtyping ones,
capturing relations between two Liquid Intersection Types. A
constraint $\Gamma\vdash^\cap\sigma\prec\sigma'$ is valid if the type 
$\sigma'$ is a super-type of $\sigma$, meaning that there is a type
derivation using the subsumption rule to relate the two types.

The well-formedness and subtyping rules (Figure \ref{fig:type-system})
can be used to simplify constraints prior to their solving. For
instance, the constraint
$\Gamma\vdash^\cap\tau_1\cap\ldots\cap\tau_n$ can 
be simplified to the set
$\left\{\Gamma\vdash^\cap\tau_1,\ldots,\Gamma\vdash^\cap\tau_n\right\}$.
On the other hand, the constraint $\Gamma\vdash^\cap 
(x:\tau_1\rightarrow\tau_2)\prec 
(x:\tau_1'\rightarrow\tau_2')$ can be further reduced to
$\Gamma\vdash^\cap\tau_1'\prec\tau_1$ and
$\Gamma;x:\tau_1'\vdash^\cap\tau_2\prec\tau_2'$. 

\subsection{Constraint solving} 
\label{sec:cons-validity}

We now describe the process of solving the collected constraints
throughout the inference algorithm. This process will reduce to two
different validity tests: a well-formedness constraint will,
ultimately, reduce to the constraint of the form
$\Gamma\vdash^\cap\left\{\nu : B \: | \: E\right\}$ and so it will
amount to check if the type \texttt{bool} can be derived for $E$
under $\Gamma$; for the subtyping case, the simplification of
constraints will result in a series of restrictions of the form
$\Gamma\vdash^\cap\left\{\nu : B \: | \:
  E_1\right\}\cap\ldots\cap\left\{\nu : B \: | \:
  E_n\right\}\prec\left\{\nu : 
B \: | \: E_1'\right\}\cap\ldots\cap\left\{\nu : B \: | \: 
  E_m'\right\}$, leading to check if
$\left\llbracket\Gamma\right\rrbracket \wedge
\left\llbracket E_1\right\rrbracket\wedge\ldots
\wedge\left\llbracket E_n\right\rrbracket\Rightarrow
\left\llbracket E_1'\right\rrbracket\wedge\ldots 
\wedge\left\llbracket E'_m\right\rrbracket$ holds.

Whenever well-formedness constraints are generated, these are solved
before the subtyping ones. This step ensures only well-formed types
are involved in subtyping relations. Well-formedness constraints arise
when a \emph{fresh} Liquid Intersection Type is generated, since that
is when refinement expressions are plugged into a type. 
Such fresh types will be of the form $\tau_1\cap\ldots\cap\tau_n$,
so the solution for a constraint of the form
$\Gamma\vdash^\cap\tau_1\cap\ldots\cap\tau_n$ is the type
$\bigcap\left\{\tau_i\right\}$, the intersection of all $\tau_i$
(with $1\leq i\leq n$) such that $\Gamma\vdash^\cap\tau_i$. We
assign this solution to a \emph{temporary} type, denoted by
$\mathcal{A}$, which will be used during the solving of subtyping
constraints. 

The subtyping constraints will ensure that inferred types only
present refinement expressions capturing the functional behavior of
terms. These will be used with $\lambda$-abstractions, applications
and let-bindings. Except for applications, subtyping constraints are
preceded by the resolution of well-formedness restrictions, and so it
is the case that subtyping relations will be checked using the
temporary type $\mathcal{A}$. 

For the case of $\lambda$-abstractions,
after generating the fresh Liquid Intersection Type
$(x:\hat{\tau}_1\rightarrow\hat{\tau}_1')\cap\ldots\cap
(x:\hat{\tau}_n\rightarrow\hat{\tau}_n')$, a series of calls 
to $\mathtt{Infer}$ are triggered, which we present via the syntax
$\mathrm{let}\;\tau_i'' = \mathtt{Infer}(\Gamma;x:\hat{\tau}_i, M,
\mathbb{Q})$, with $1\leq i \leq n$. These calls differ only on the
type $\hat{\tau}_i$ of $x$ pushed into the environment, implying that
different types for $M$ can be inferred. After solving the
well-formedness constraints, we must remove from type $\mathcal{A}$
the refinement expressions that would cause the type to be
unsound. We use the notation
$x:\tau_k\rightarrow\tau_k'\in\mathcal{A}$  to indicate that
$\bigcap\left\{x:\tau_k\rightarrow\tau_k'\right\}$ should be a
supertype of $\mathcal{A}$, in the sense that it can be obtained from
$\mathcal{A}$ using exclusively the rule [{\sc $\prec$-Elim}] (taking
an analogy with set theory, 
$\bigcap\left\{x:\tau_k\rightarrow\tau_k'\right\}$ would be a
sub set of the intersections of $\mathcal{A}$). Then, the inferred
type will be
$\bigcap\left\{x:\hat{\tau}_k\rightarrow\hat{\tau}_k'\right\}$, 
such that $x:\hat{\tau}_k\rightarrow\hat{\tau}_k'\in\mathcal{A}$ and the
constraint
$\Gamma;x:\hat{\tau}_k\vdash^\cap \tau_k''\prec\hat{\tau}_k'$ is
valid, 
that is, the type inferred for $M$ under the environment
$\Gamma;x:\hat{\tau}_k$ is a subtype of $\hat{\tau}_k'$. As an
example, consider 
$\mathbb{Q} = \left\{\nu\geq 0, \nu\leq 0, y = 5\right\}$, the term
$\lambda x. -x$ and $\Gamma=\emptyset$. The inference procedure will
start by generating the type:
\vspace{-2.5pt}
\[
\begin{array}{l}
(x:\left\{\nu\geq 0\right\}\rightarrow\left\{\nu \geq 0\right\}) \:
\cap \\ 
(x:\left\{\nu\geq 0\right\}\rightarrow\left\{\nu \leq 0\right\}) \:
\cap \\  
(x:\left\{\nu\leq 0\right\}\rightarrow\left\{\nu \geq 0\right\}) \:
\cap \\
(x:\left\{\nu\leq 0\right\}\rightarrow\left\{\nu \leq 0\right\}) \:
\cap \\
(x:\left\{\nu \geq 0\right\}\rightarrow\left\{y = 5\right\}) \: \cap
\\
(x:\left\{\nu \leq 0\right\}\rightarrow\left\{y = 5\right\}) \: \cap
\\
(x:\left\{y = 5\right\}\rightarrow\left\{\nu \geq 0\right\}) \: \cap
\\
(x:\left\{y = 5\right\}\rightarrow\left\{\nu \leq 0\right\}) \: \cap
\\
(x:\left\{y = 5\right\}\rightarrow\left\{y = 5\right\})
\end{array}
\]
\vspace{-2.5pt}
Then, with well-formedness constraints, and since no variable $y$ is
in scope, we are left with:
\vspace{-2.5pt}
\[
\begin{array}{l}
(x:\left\{\nu\geq 0\right\}\rightarrow\left\{\nu \geq 0\right\}) \:
\cap \\ 
(x:\left\{\nu\geq 0\right\}\rightarrow\left\{\nu \leq 0\right\}) \:
\cap \\  
(x:\left\{\nu\leq 0\right\}\rightarrow\left\{\nu \geq 0\right\}) \:
\cap \\
(x:\left\{\nu\leq 0\right\}\rightarrow\left\{\nu \leq 0\right\})
\end{array}
\] 
\vspace{-2.5pt}
Finally, because of subtyping
relations, the inferred type will be:
\vspace{-2.5pt}
\[
\begin{array}{l}
(x:\left\{\nu\geq 0\right\}\rightarrow\left\{\nu \leq 0\right\}) \:
\cap \\  
(x:\left\{\nu\leq 0\right\}\rightarrow\left\{\nu \geq 0\right\})
\end{array}
\]

For application and let-bindings, solving subtyping constraints works
in a similar manner as for $\lambda$-abstractions. The type of an
application is inferred similarly as in
\cite{Freeman:1991:RTM:113445.113468}: for the function $M$ with type
$x:\tau_1\rightarrow\tau_1'\cap\ldots\cap\tau_n\rightarrow\tau_n'$ 
and the argument $N$ with type $\tau$, the type of $MN$ is
$\bigcap\left\{\tau_i'\right\}$, such that $1\leq i \leq n$ and
$\Gamma\vdash^\cap\tau\prec\tau_i$ is checked valid.

\subsection{Properties of inference}
\label{sec:proper-infer}

We were able to prove that our inference algorithm is sound with
respect to the typing rules. 

\begin{lemma}[Relation with derivation and well-founded types] If \hspace{.5pt}
  $\Gamma\vdash^\cap_{\mathbb{Q}} M : \sigma$ then $\sigma::\mathtt{Shape}(\sigma)$.
\label{lemma:wf-d}
\end{lemma}

\begin{proof}By straightforward induction over
  $\Gamma\vdash^\cap_{\mathbb{Q}} M : \sigma$.\end{proof}

\vspace{3pt}
\begin{theorem}[Soundness of inference] If\hspace{2.5pt} $\mathrm{Infer}(\Gamma, M, \mathbb{Q})=
  \sigma$, then $\judgment{M}{\sigma}$.
\end{theorem}

\begin{proof} By structural induction over $M$.

  \begin{itemize}
  \item case $M\equiv x$:
    \begin{itemize}
    \item subcase in which $M$ has a basic type
      in this case $\mathcal{W}(\mathtt{Shape}(\Gamma), x) = B$ and so
      $x$ has type $\left\{\nu : B \: | \:
        \phi_1\right\}\cap\ldots\cap\left\{\nu : B \: | \:
        \phi_n\right\}$, which we abbreviate to
      $\tau_1\cap\cdots\cap\tau_n$.
      
      The following derivation is then valid
      \[
      \inferrule*[Right=\rmfamily\scshape B-Var]
      {
        \Gamma(x) = \tau_1\cap\cdots\cap\tau_n \\
        \tau_i :: B (\forall i. 1 \leq i \leq n)
      }
      {
        \Gamma\vdash^{\cap}_{\mathbb{Q}} x : \left\{\nu : B \: | \:
          \nu = x \right\} 
      }
      \]
    \item subcase in which $x$ has not a basic type: 
      in this case $\sigma=\Gamma(x)$.
      
      So, the following derivation is valid
      \[
      \inferrule*[Right=\rmfamily\scshape Var]
      {
        \Gamma(x) = \sigma \\
        \Gamma(x) :: \mathtt{Shape}(\sigma)
      }
      {
        \Gamma\vdash^{\cap}_{\mathbb{Q}} x : \sigma
      }
      \]
    \end{itemize}
  \item Case $M\equiv c$: Easy, by application of the rule
    [{\scshape Const}].
  \item Case $M\equiv \lambda x.N$: In this case the algorithm computes
    \begin{itemize}
    \item
      $(x:\hat{\tau}_1\rightarrow\hat{\tau}'_1)\cap\ldots
      \cap(x:\hat{\tau}_n\rightarrow\hat{\tau}'_n) 
      = \texttt{Fresh}(\mathcal{W}(\texttt{Shape}(\Gamma), \lambda
      x.M), \mathbb{Q})$ 
    \end{itemize}
    By IH
    \begin{align}
    & \Gamma;x:\hat{\tau}_i\vdash^{\cap}_{\mathbb{Q}} N : \tau_i'',
    \; \forall i : 1 \leq i \leq n
    \tag{a}
    \label{eq:a}
    \end{align}
    
    By Lemma~\ref{lemma:wf-d}
    \[
    \tau''_i::\mathtt{Shape}(\tau''_i), \forall i : 1 \leq i \leq n
    \]

    The type $\mathcal{A}$ restricts the inferred type only to the
    well formed intersections:
    $\Gamma\vdash^{\cap}(x:\hat{\tau}_1\rightarrow\hat{\tau}'_1)
    \cap\ldots\cap(x:\hat{\tau}_n\rightarrow\hat{\tau}'_n)$   
    reduces to:
    $$\left\{\Gamma\vdash^{\cap}
      (x:\hat{\tau}_1\rightarrow\hat{\tau}'_1),\ldots,
      \Gamma\vdash^{\cap}(x:\hat{\tau}_n\rightarrow\hat{\tau}'_n)
    \right\}$$
    
    Consider the sub-set of derivations in (\ref{eq:a}) such that
    $\Gamma;
    x:\hat{\tau}_j\vdash^{\cap}\tau_j''\prec\hat{\tau}'_j$
    and that respects the type $\mathcal{A}$. We can conclude that
    $\hat{\tau}_j::Shape(\tau''_j)$ as the subtyping relation can be
    only applied to types refining the same ML type. We shall use $T$
    to denote $\mathtt{Shape}(\tau''_j)$.
    
    We have then a set of derivations of the form 
    \[
    \small
    \inferrule*[Left=\footnotesize\rmfamily\scshape Fun]
    {
      \inferrule*[Left=\footnotesize\rmfamily\scshape Sub]
      {
        \Gamma;x:\hat{\tau}_j\vdash^{\cap}_{\mathbb{Q}} N : \tau_j'' \\ 
        \Gamma;x:\hat{\tau}_j\vdash^{\cap} \tau_j'' \prec
        \hat{\tau}_j' \\ 
        \Gamma;x:\hat{\tau}_j \vdash^{\cap} \hat{\tau}'_j
      }
      {
        \Gamma;x:\hat{\tau}'_j\vdash^{\cap}_{\mathbb{Q}} N : \hat{\tau}'_j 
      } \\ \Gamma\vdash^{\cap} x: \hat{\tau}_j\rightarrow\hat{\tau}'_j 
        \\ \hat{\tau}'_j :: T
    }
    {
      \Gamma\vdash^{\cap}_{\mathbb{Q}}\lambda x.N :
      (x:\hat{\tau}_j\rightarrow\hat{\tau}'_j) 
    }
    \]

    By Lemma~\ref{lemma:wf-d}
    \[
    x:\hat{\tau_j}\rightarrow\hat{\tau_j'} ::
    \mathtt{Shape}(\hat{\tau_j})\rightarrow\mathtt{Shape}(\hat{\tau_j'})
    \]
    Moreover, $\mathtt{Shape}(\hat{\tau_j'}) = T$ and we shall we use
    $T'$ to denote $\mathtt{Shape}(\hat{\tau_j})$.
    
    By repeated application of the rule [{\scshape Intersect}]
    \[
    \inferrule*[Left=\footnotesize\rmfamily\scshape Intersect]
    {
      (x:\hat{\tau}_j\rightarrow\hat{\tau}'_j)\cap\ldots
      \cap(x:\hat{\tau}_{j+k}\rightarrow\hat{\tau}'_{j+k}) :: T'\rightarrow T \\\\
      \Gamma\vdash^{\cap}_{\mathbb{Q}}\lambda x.N :
      (x:\hat{\tau}_j\rightarrow\hat{\tau}'_j) \\ \ldots \\ 
      \Gamma\vdash^{\cap}_{\mathbb{Q}}\lambda x.N :
      (x:\hat{\tau}_{j+k}\rightarrow\hat{\tau}'_{j+k})
    }
    {
      \Gamma\vdash^{\cap}_{\mathbb{Q}}\lambda x.N:
      (x:\hat{\tau}_j\rightarrow\hat{\tau}'_j)\cap\ldots
      \cap(x:\hat{\tau}_{j+k}\rightarrow\hat{\tau}'_{j+k})
    }
    \]
    
  \item case $M\equiv M'N$: By IH
    \begin{itemize}
    \item 
      $
      \Gamma\vdash^{\cap}_{\mathbb{Q}} M' :
      (x:\tau_1\rightarrow\tau_1')
      \cap\ldots\cap(x:\tau_n\rightarrow\tau_n')  
      $
    \item $\Gamma\vdash^{\cap}_{\mathbb{Q}} N : \tau$
    \end{itemize}

    Consider $\mathcal{D}$ the following derivation
    \[
    \inferrule*[Right=Sub]
      {
        \Gamma\vdash^{\cap}_{\mathbb{Q}} N : \tau \\
        \Gamma\vdash^{\cap} \tau\prec \tau_i \\
        \Gamma\vdash^{\cap}\tau_i
      }
      {
        \Gamma\vdash^{\cap}_{\mathbb{Q}} N : \tau_i
      }
    \]

    For all the $\tau_i$ such that $\tau\prec\tau_i$ we
    have a derivation of the form
    \[
    \small
    \inferrule*[Right=\scriptsize \rmfamily\scshape App]
    {
      \inferrule*[Left=\scriptsize \rmfamily\scshape Sub]
      {
        \Gamma\vdash^{\cap}_{\mathbb{Q}} M' :
        (x:\tau_1\rightarrow\tau_1')\cap
        \ldots\cap(x:\tau_n\rightarrow\tau_n') 
        \\ 
        \Gamma\vdash^{\cap} (x:\tau_1\rightarrow\tau_1')\cap
        \ldots\cap(x:\tau_n\rightarrow\tau_n') \prec
        (x:\tau_i\rightarrow\tau_i') \\
        \Gamma\vdash (x : \tau_i\rightarrow\tau_i')
      }
      {
        \Gamma\vdash^{\cap}_{\mathbb{Q}} M' : (x:\tau_i\rightarrow\tau_i') 
      } \\
      \mathcal{D}
    }
    {
      \Gamma\vdash^{\cap}_{\mathbb{Q}} M'N : \tau_i'[N/x]
    }
    \]

    Let $\mathcal{D}_1$ be the previous derivation.
    For each $\tau_i$ that satisfy $\tau\prec\tau_i$
    we have a derivation of the previous form.

    By Lemma~\ref{lemma:wf-d}
    \[
    \tau'_i[N/x] :: \mathtt{Shape}(\tau'_i[N/x])
    \]
    and we shall use $T$ to denote $\mathtt{Shape}(\tau'_i[N/x])$.
    So, by repeated application of the rule [{\scshape Intersect}] the
    following derivation is valid
    \[
    \inferrule*[Right=\rmfamily\scshape Intersect]
    {
      \mathcal{D}_i \\ \ldots \\ \mathcal{D}_{i+j} \\
      \tau_i'[N/x]\cap\ldots\cap\tau_{i+j}'[N/x] :: T
    }
    {
      \Gamma\vdash^{\cap}_{\mathbb{Q}} M'N :
      \tau_i'[N/x]\cap\ldots\cap\tau_{i+j}'[N/x] 
    }
    \]

    By the definition of substitution we have
    $\tau_i'[N/x]\cap\ldots\cap\tau_{i+j}'[N/x] =
    (\tau_i'\cap\ldots\cap\tau_{i+j}')[N/x]$, which is precisely the
    inferred type.

  \item case $M\equiv \mathrm{let}\,x= M' \,\mathrm{in} \,N$:
    $\sigma$ is of the form $\hat{\tau}_1''\cap\ldots\cap\hat{\tau}_n''$. 
    By IH
    \begin{itemize}
    \item $\Gamma\vdash^{\cap}_{\mathbb{Q}} M' : \tau_1$
    \item $\Gamma;x:\tau_1\vdash^{\cap}_{\mathbb{Q}} N : \tau_2$
    \end{itemize}

    The type $\mathcal{A}$ stands for the set of $\hat{\tau}_i$ such
    that $\Gamma\vdash^{\cap}\hat{\tau}_i$, which by the definition
    of well formed type we have
    \begin{align}
      & \inferrule*[Left=WF-Intersect]
      {
        \Gamma\vdash^{\cap}\hat{\tau}_{{i}_1} \\ 
        \ldots \\
        \Gamma\vdash^{\cap}\hat{\tau}'_{{i}_n}
      }
      {
        \Gamma\vdash^{\cap}\hat{\tau}_{{i}_1}'\cap\ldots\cap\hat{\tau}'_{{i}_n}
      }
      \tag{b}
      \label{eq:b}
    \end{align}
    
    Now we consider all $\hat{\tau}_j$ in $\mathcal{A}$ such that
    $\Gamma;x :
    \tau_1\vdash^\cap\tau_2\prec\hat{\tau_j}$. We have
    that $\Gamma\vdash^\cap\hat{\tau_j}$ as this is a type taken from
    $\mathcal{A}$.
    We then have a series of derivations of the form 
    \[
    \inferrule*[Left=\scriptsize Sub]
    {
      \Gamma;x:\tau_1\vdash^{\cap}_{\mathbb{Q}} N : \tau_2 \\
      \Gamma;x:\tau_1\vdash^{\cap}
      \tau_2\prec\hat{\tau}_j \\
      \Gamma\vdashª^\cap \hat{\tau}_j
    }
    {
      \Gamma;x:\tau_1\vdash^{\cap}_{\mathbb{Q}} N : \hat{\tau}_j
    }
    \]

    By Lemma~\ref{lemma:wf-d} 
    \[
    \hat{\tau}_j :: \mathtt{Shape}(\hat{\tau}_j)
    \]
    and we will use $T$ for $\mathtt{Shape}(\hat{\tau}_j)$.
    By repeated application of the rule [{\scshape Intersect}]
    \[
    \inferrule*[Left=Intersect]
    {
      \Gamma;x:\tau_1\vdash^{\cap}_{\mathbb{Q}} N : \hat{\tau}_{{j}_1} \\
      \ldots \\ \Gamma;x:\tau_1\vdash^{\cap}_{\mathbb{Q}} N : \hat{\tau}_{{j}_k} \\
      \hat{\tau}_{{j}_1} \cap \ldots \cap \hat{\tau}_{{j}_k} :: T
    }
    {
      \Gamma;x:\tau_1\vdash^\cap_{\mathbb{Q}} N : \hat{\tau}_{{j}_1} \cap \ldots \cap \hat{\tau}_{{j}_k}
    }
    \]

    The following derivation is then valid
    \[
    \fontsize{9.8pt}{1em}\selectfont
    \inferrule*[Left=\scriptsize Let]
    {
      \Gamma\vdash^{\cap}_{\mathbb{Q}} M' : \tau_1 \\
      \Gamma;x:\tau_1\vdash^{\cap}_{\mathbb{Q}} N :
      \hat{\tau}_{{j}_1}\cap\ldots\cap\hat{\tau}_{{j}_k}  
      \\
      \inferrule*[Right=(c)]
      {
        \ddots
      }
      {
        \Gamma\vdash^{\cap}_{\mathbb{Q}}\hat{\tau}_{{j}_1}\cap
        \ldots\cap\hat{\tau}_{{j}_k}  
      }
    }
    {
      \Gamma\vdash^{\cap}_{\mathbb{Q}} \mathrm{let}\;x = M' \;
      \mathrm{in} \; N : \hat{\tau}_{{j}_1}\cap\ldots\cap\hat{\tau}_{{j}_k} 
    }
    \]

    The derivation (c) follows by (\ref{eq:b}), since it is the exact
    same derivations but now we only consider the $\hat{\tau}_j$ such
    that $\Gamma;x :
    \tau_1\vdash^\cap_{\mathbb{Q}}\tau_2\prec\hat{\tau_j}$, i.e. we
    intersect a sub-set of the types in (\ref{eq:b}).

  \item case $M\equiv[\Lambda\alpha]M'$: By IH
    \[
    \Gamma\vdash^{\cap}_{\mathbb{Q}} M' : \sigma
    \]

    The following derivation is valid
    \[
    \inferrule*[Right=Gen]
    {
      \Gamma\vdash^{\cap}_{\mathbb{Q}} M' : \sigma \\
      \alpha\not\in\Gamma
    }
    {
      \Gamma\vdash^{\cap}_{\mathbb{Q}}M' : \forall\alpha.\sigma
    }
    \]

  \item case $M\equiv[\tau]M'$: By IH
    \[
    \Gamma\vdash^{\cap}_{\mathbb{Q}} M' : \forall\alpha.\sigma
    \]

    Since 
    $\tau' = \mathtt{Fresh}(T, \mathbb{Q})$, then
    $T=\mathtt{Shape}(\tau')$. 

    $\tau'$ is of the form $\tau'_1\cap\ldots\cap\tau'_n$.
    The type
    $\mathcal{A}$ stands for the set of all 
    $\tau'_i$ such that $\Gamma\vdash^{\cap}\tau'_i$, so it is a
    sub-type of $\tau'_1\cap\ldots\cap\tau'_n$. 
    Then, the following derivation is valid
    \[
    \small
    \inferrule*[Left=\scriptsize Inst]
    {
      \Gamma\vdash^{\cap}_{\mathbb{Q}} M' : \forall\alpha.\sigma \\
      \inferrule*[lab=\scriptsize WF-Intersect]
      {
        \Gamma\vdash^{\cap}\tau'_i \\
        \ldots \\
        \Gamma\vdash^{\cap}\tau'_{i+j}
      }
      {
        \Gamma\vdash^{\cap}\tau'_i\cap\ldots\cap\tau'_{i+j}
      }\\
      \mathtt{Shape}(\tau'_i\cap\ldots\cap\tau'_{i+j}) = T
    }
    {
      \Gamma\vdash^{\cap}_{\mathbb{Q}}[\tau]M' :
      \sigma[\tau'_i\cap\ldots\cap\tau'_{i+j}/\alpha] 
    }
    \]
  \end{itemize} 
  \qed
\end{proof}

\subsection{The \emph{lisette} tool}
\label{sec:emphlisette-tool}

In order to automate all the \emph{proof-and-typing} process required
for Liquid Intersection Types inference, we implemented a prototype
tool that we baptized \textbf{lisette} (LIquid interSEction
TypEs)\footnote{\url{http://www.dcc.fc.up.pt/~mariopereira/lisette.tar.gz}}.

The purpose of \textbf{lisette} is to parse a program written in a
ML-like language (which we shall designate \emph{tiny-ML}) plus a set
of logical qualifiers and infer an appropriate Liquid Intersection
Type for that program, requiring no further assistance from the
user. This tool works as follows:
\begin{enumerate}
\item \textbf{lisette} parses the tiny-ML file (program plus
  qualifiers) and produces its A-normal form version;
\item using Damas-Milner inference engine, an ML type is computed for
  each sub-term in the program;
\item using the $\mathsf{Fresh}(\cdot,\cdot)$ function, the Liquid
  Intersection Type containing all possible combinations of qualifiers
  is generated and assigned to each sub-term;
\item then, depending on which term is being processed, a set of
  well-formedness constraints are generated, solved by testing if for
  all refinement expressions the type \textit{bool} can be derived;
\item to respect the relations between types, a set of subtyping
  constraints is computed and translated to an equivalent logical
  formula;
\item using the logic of the Why3 platform~\cite{filliatre13esop,
    filliatre13cade} as a back-end, we use several automatic theorem
  provers to test the validity of the generated subtyping constraints;
\item finally, combining the results of solving well-formedness and
  subtyping constraints, the final Liquid Intersection Type is
  assigned to the corresponding sub-term.
\end{enumerate}

Our use of the Why3 platform API is motivated by the fact that its
internal logic can target multiple provers. This allows the user of
\textbf{lisette} to experiment with different provers, comparing how
well they perform in solving the generated constraints. If the user
does not specify a particular prover to be used, then \textbf{lisette}
tries to solve a constraint by using all the available provers,
stopping with the first one that is able to prove the validity of the
constraint. If none returns a positive answer, that constraint is
marked as false. Another advantage of using Why3 is that when
designing the tool there is no need to worry about the different input 
languages of each different prover, being enough to implement a single
translation function from the language of Liquid Intersection Types to
Why3 terms.

\begin{figure}
\begin{center}
  \begin{tabular}{c}
\begin{minipage}{0.25\textwidth}
\begin{verbatim}
Qualifiers
{
   v >= 0,	
   v <= 0
}

val mul = \x . * x x
val neg = \x. - x
\end{verbatim}
\end{minipage}
\end{tabular}
\end{center}
\caption{File accepted by the \textbf{lisette} tool: a set of logical
  qualifiers and a program written in tiny-ML.}
\label{fig:tiny-ml}
\end{figure}

As mentioned, this tool accepts a file
containing a set of logical qualifiers and a program written in
tiny-ML, such as the one in Figure~\ref{fig:tiny-ml}.
For this example we have $\mathbb{Q} = \{\nu \geq 0, \nu \leq 0\}$ and
the terms composing the program are $\mathit{neg}\equiv\lambda x. -x$ and
$\mathit{mul}\equiv\lambda x. *x\;x$. Using the supplied set,
\textbf{lisette} will produce the following output:  
\begin{lstlisting}[flexiblecolumns=true]
-----------------------------------
Inference result:       
      mul : (x: {v : int | (v>=0)} -> {v : int | (v>=0)}) /\ 
            (x: {v : int | (v<=0)} -> {v : int | (v>=0)})

      neg : (x: {v : int | (v<=0)} -> {v : int | (v>=0)}) /\ 
            (x: {v : int | (v>=0)} -> {v : int | (v<=0)}) 
-----------------------------------
\end{lstlisting}
At the end, \textbf{lisette} is able to infer sound and expressive
Liquid Intersection Types for the terms $\mathit{mul}$ and
$\mathit{neg}$.

 

\section{Conclusion and future work}
\label{sec:conclusion}

We presented a new type system supporting functional descriptions, via
refinement types, and offering the expressiveness of intersection
types. We believe our type system can be used to derive more precise
types than previous refinement type systems, whilst maintaining
type-checking and inference decidable. Liquid
Types~\cite{Rondon:2008:LT:1375581.1375602} tend to infer poorly
accurate and even meaningless refinement types for some terms (leading
to the absence of principal types), which we preclude due to the
precision of intersection in types. Refinement types for algebraic
data-types~\cite{Freeman:1991:RTM:113445.113468} are precise and
present desirable properties such as principality and decidable
inference, though it is our believe that logical predicates are a more
natural way to specify functional behavior of programs. General
refinement types~\cite{Knowles:2010:HTC:1667048.1667051} use a very
expressive annotations language, allowing to assign very precise types
to programs, yet with the serious drawback of undecidable
type-checking and inference. With Liquid Intersection Types we
maintain our predicates language simple, while being able to
automatically infer very accurate and meaningful refinement types.

To design a decidable system we adopted a style closely related to
Liquid Types: the refinement expressions presented in types are
exclusively collected from $\mathbb{Q}$, a global set of logical
qualifiers, and the subtyping is decidable. We also impose that the
type of an expression must the intersection of refinements to its ML
type, intersecting only types of the same form.

We also proposed an inference algorithm for Liquid Intersection
Types. This algorithm takes as input an environment $\Gamma$, a term
$M$ and the set of qualifiers $\mathbb{Q}$, producing the
correspondent Liquid Intersection Type. Our inference algorithm uses
the $\mathcal{W}$ algorithm to infer the shape of a Liquid
Intersection Type, which is the ML type for that term. To determine
which refinement expressions can be plugged into a type, the algorithm
produces a series of well-formedness and subtyping constraints,
solving them immediately after their generation. We have been able to
prove that our algorithm is sound with respect to the conceived typing
rules.


Current and future work includes the study of completeness of type
inference for our system and to extend decidable intersection type
systems (of finite ranks
\cite{jim1995rank,Kfoury:1999:PDT:292540.292556}) with type refinement
predicates.



\bibliographystyle{eptcs}
\bibliography{refs}

\end{document}